\documentclass[
  aps,              
  prx,              
  reprint,
  superscriptaddress,
]{revtex4-2}

\usepackage[T1]{fontenc}
\usepackage[utf8]{inputenc}
\usepackage{lmodern}
\usepackage[dvipsnames]{xcolor}
\usepackage{microtype}           
\usepackage{graphicx}            
\usepackage{amsmath,amssymb}
\usepackage{mathtools}
\usepackage{siunitx}
\usepackage{physics}             
\usepackage{xspace}
\usepackage{pgfplots}
\usepackage{graphicx}
\usepackage{fancybox}
\usepackage{tikz}
\usepackage[colorlinks]{hyperref}
\usepackage{comment}
\hypersetup{
    colorlinks  = true,
    citecolor   = PineGreen,
    linkcolor   = MidnightBlue,
    urlcolor    = PineGreen
}

\usepackage{amsthm}
\usepackage{thmtools}

\declaretheorem[numbered=yes]{theorem}

\usepackage[capitalize,nameinlink,noabbrev]{cleveref}

\pgfmathdeclarefunction{hbin}{1}{%
  \pgfmathsetmacro{\qclamp}{max(min(#1,1-1e-12),1e-12)}%
  \pgfmathparse{-(\qclamp)*ln(\qclamp)/ln(2) - (1-\qclamp)*ln(1-\qclamp)/ln(2)}%
}

\pgfmathdeclarefunction{Ranalytic}{2}{%
  \pgfmathsetmacro{\v}{max(min(#1,1),0)}%
  \pgfmathsetmacro{\vloc}{max(min(#2,1),0)}%
  \pgfmathsetmacro{\q}{(1-\v)/2}%
  \pgfmathparse{%
    1
    - ln(1 + 2*\vloc*sqrt(max(0,1-\vloc*\vloc)))/ln(2)
    - 2*hbin(\q)
  }%
}

\pgfmathdeclarefunction{RanalyticPlus}{2}{%
  \pgfmathparse{max(0, Ranalytic(#1,#2))}%
}

\pgfmathdeclarefunction{RSPplus}{1}{%
  \pgfmathsetmacro{\v}{max(min(#1,1),0)}%
  \pgfmathsetmacro{\q}{(1-\v)/2}%
  \pgfmathparse{max(0, 1 - 2*hbin(\q))}%
}

\pgfmathdeclarefunction{RPironioPlus}{1}{%
  \pgfmathsetmacro{\v}{max(min(#1,1),0)}%
  \pgfmathsetmacro{\q}{(1-\v)/2}%
  \pgfmathsetmacro{\rad}{max(0, 2*\v*\v - 1)}%
  \pgfmathsetmacro{\pE}{(1 + sqrt(\rad))/2}%
  \pgfmathparse{max(0, (\v >= 1/sqrt(2))*(1 - hbin(\q) - hbin(\pE)))}%
}

\begin{document}

\title{Routed Bell tests with arbitrary many local parties}

\author{Gereon Ko\ss mann}
\affiliation{Institute for Quantum Information, RWTH Aachen University, Aachen, Germany}

\author{Mario Berta}
\affiliation{Institute for Quantum Information, RWTH Aachen University, Aachen, Germany}

\author{René Schwonnek}
\affiliation{Leibniz Universität Hannover, Hannover, Germany}


\date{\today}

\begin{abstract}
    Device-independent quantum key distribution (DIQKD) promises cryptographic security based solely on observed quantum correlations, yet its implementation over long distances remains limited. Routed Bell tests have recently re-emerged as a promising strategy to mitigate this limitation by enabling local self-testing of one party’s device. However, extending this idea to self-testing both communicating parties has remained unclear.

    Here we develop a general $C^*$-algebraic for routed DIQKD with multiple switches and arbitrarily many local test parties, with a conservative, state-dependent definition of Eve. Within this framework, we design and analyse four-party routed protocols that locally self-test both Alice as well as Bob, and numerically bound key rates from the full observed statistics. In the parameter regimes considered, adding a fourth party strictly improves certified key rates and lowers the non-zero key threshold. Randomized key-basis switching further amplifies this advantage. Finally, we investigate a self-test–assisted E91-type protocol that continuously interpolates between the device-dependent Shor–Preskill rate and the device-independent rate derived by [Pironio {\it et al.}, New J.~Phys.~11, 045021 (2009)].
\end{abstract}

\maketitle



\section{Introduction}\label{sec:intro}

When connected by a quantum channel of sufficient quality, two distant parties can create an information theoretically secure key \cite{Shannon2001,RENNER2008} by performing a Quantum key distribution (QKD) protocol \cite{Pirandola_2020,bbm92,Bennett1992}. The basic intuition underlying security proofs is rooted in fundamental principles of quantum physics \cite{Ekert2014}: Any eavesdropping attempt necessarily disturbs the quantum channel, which leaves a fingerprint that can be detected by the data generated during the protocol runs.
In practice, turning observed data into a security guarantee does, however, not only depend on the noise in the channel. It also hinges on how we turn the real world devices, used  at the endpoints of a channel, into a mathematical model.

In any such model, one must specify which degrees of freedom are relevant, how the devices act on them, and which imperfections are included or idealized. This leads to a natural trade-off between the strength of the assumptions imposed and the validity of a security guarantee for real world implementation. Under strong device-dependent assumptions, QKD has matured into a technology with long-distance, real-world deployments and commercial products. Device-independent (DI) QKD  \cite{Mayers,zapatero2023advances,primaatmaja2023security}, by contrast, aims for security statements that keep device assumptions at a minimum. In a nutshell, these are isolated laboratories and freely chosen measurement settings \cite{primaatmaja2023security}. As such a security guarantee comes close to rely on observed input-output statistics alone. Nevertheless, this comes at the cost of imposing high technological and mathematical hurdles. It took roughly three decades of work since the first broadly accepted demonstrations where performed. However, those proofs of principles are still limited, for instance to short distances \cite{Nadlinger2022,liu2022toward} or low rates \cite{Zhang2022,lu2026device}.   
\begin{figure}
    \centering
    \includegraphics[width=1\linewidth]{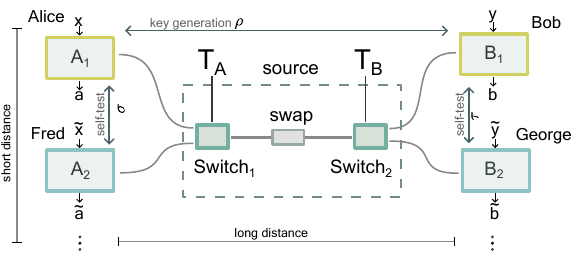}
    \caption{Two-sided routed architecture combining local Bell tests with several parties. In each round, local routing variables $T_A$ and $T_B$ on switches decide on the distribution of states for the local and long-range link. Two switches can be combined by entanglement swapping \cite{Zukowski1993} or via event-ready schemes employing a quantum relay.}
    \label{fig:scheme_two_switches}
\end{figure}

In this work, we study protocol architectures that promise to bridge this gap. Ideally, they retain the device-independence of the security claim while aiming for performance close to device-dependent settings.

In a sufficiently entangled quantum system measurements and states can self-test~\cite{coladangelo2017all,scarani2019bell,vsupic2020self,operator_algebra_self_test2024}. Protocols based on routed local/global Bell-test architectures (see \autoref{fig:scheme_two_switches}) employ this by placing additional local test parties inside the users’ laboratories. They can bypass the noise in a long range quantum channel and create local states of high quality that are used to self-test the relevant measurement behaviour of the endpoint devices for the key generation \cite{Lim2013,Chaturvedi2024,Lobo2024,Le_Roy_Deloison_2025,Tan_2024}. In a full protocol a switch selects the connectivity in each round. Either the endpoints are connected to each other for key generation, or each endpoint is connected to its local tester for certification. If the local self-tests are sufficiently strong, then\,---\,up to a controlled error\,---\,the long-distance key-generation statistics behave as if the endpoint measurements were trusted. In the ideal self-testing limit, one therefore expects the DI key rate to approach the corresponding device-dependent benchmark, with the gap governed by the self-testing error. We observe this behaviour in \autoref{fig:interpolation}, where we numerical plot asymptotic key rates.

This idea goes back to  works of 
Lim and Tomamichel \cite{Lim2013}. They provide a rate formula that is analytic, but in most regimes, unfortunately, suboptimal. Numerical evidence for rate improvements in relevant regimes has been reported \cite{Tan_2024,Le_Roy_Deloison_2025}.
Here we provide a general modelling and proof framework that turns this intuition into quantitative security statements and allows us to optimize protocols beyond previously accessible settings.

Existing analyses also restrict, to scenarios with only one additional local test on only one side. In particular, it remained unclear how to extend the routed approach to simultaneously certify both communicating devices within a single conservative DI security model, and how to scale it to arbitrarily many local parties and switches without relaxing the attacker model. In our work, we exactly achieve this: We formulate a fully device-independent, algebraic model for multi-party routed protocols with multiple switches and a conservative definition of Eve. Within this model, we then analyse new classes of routed four-party protocols and derive tighter key-rate bounds from the full observed statistics. In the regimes studied here, our numerics yield the best currently certified device-independent key rates for such routed architectures and identify promising candidates for future long-range experiments.


\section{Overview}

\textbf{A minimal model} -- Modelling switches is not standard. In Section~\ref{sec:model}, we introduce a minimal yet complete device-independent framework suitable for describing routed QKD protocols with multiple parties and multiple switches. In contrast to previous approaches, we model switches solely through marginal constraints, as specified in Eq.~\eqref{eq:marginal_constraint}. These constraints ensure that, even if the switches are controlled by Eve, they cannot influence devices located within the local laboratories.

Furthermore, we employ the language of $C^*$-algebras generated by projections as the proper language to express the entirety of all attacks on all quantum systems that could physically implement a protocol under investigation. While this perspective is a well-established tool in the analysis of Bell type correlations and non-local games \cite{summers1987bell,van2024schmidt,Paddock_2023}, it is only rarely used for modelling  cryptography \cite{berta2016smooth,Schwonnek2021}. In our case, this perspective  provides us with a natural notion of a state/channel dependent adversary. This avoids  technical hurdles, which previously prevented the investigation of multi switch settings. We outline in Section \ref{sec:recover} how our ansatz relates to prior models and results.\\

\textbf{Numerical Investigations} -- Deriving analytic key-rate formulas protocols with several parties and constraints from the full observed statistics seems out of reach. Consequently, we compute key rates numerically in Section \ref{sec:numerics}. To this end, we employ the method of \cite{koßmann2024boundingconditionalvonneumannentropy}, which reduces the entropy minimization to a non-commutative polynomial optimization that can be approximated via semidefinite programming \cite{Navascus2008}. In the one-sided device-independent setting (Alice characterized), it was observed that additionally characterizing Bob does not improve key rates \cite{tan2021computing}. This might suggest that tests on Bob's system are generally not  helpful. We show that this intuition does not hold in routed four-party protocols. For an extension of DI--BB84 \cite{Tan_2024} we find that adding a second local self-test on Bob's side strictly improves rates and lowers the threshold for non-zero key (see \autoref{fig:keyrate-qber-switches}). Moreover, within the same routed protocol family we introduce random key-basis switching \cite{schwonnek2020robust} on Bob's side, i.e., Bob alternates between two key-generation measurements.
This refinement of the previous scheme leads to increased certified rates and an even lower zero-key threshold in our numerics (see \autoref{fig:routed-random-key-basis}). Finally, we consider a variant of \cite{pironio2009device} enhanced by local self-tests. We find that the error in self-testing gives an interpolation between an optimal device-dependent and an optimal device-independent performance (see \autoref{fig:interpolation}). As the local self-test becomes ideal we approach Shor--Preskill key rate \cite{Shor_2000} of  BB84-type protocols, while for weak local self-testing we recover the device-independent performance of \cite{pironio2009device}. 


\section{Results}\label{sec:results}

\subsection{Modelling of a multi switch setting with many local tests}\label{sec:model}

The basic setup considered throughout this work is depicted in \autoref{fig:scheme_two_switches}. We have two sets of parties $(A_1,\dots, A_n)$ and $(B_1,\dots, B_n)$ at distant locations. The distance between a party $A_i$ and a party $B_j$ is large and the ability of creating high quality quantum states between them is limited. In contrast, all parties $\{A_i\}$ and $\{B_j\}$, respectively, are close to each other and are capable of performing a high quality Bell correlation experiment.

In order to keep notations on a convenient level we will focus in the following on describing a setting with  two switches and only four parties called Alice, Bob, Fred, and George. It will, however, become clear that nothing prevents a straightforward generalization to arbitrarily many parties on each side. Alice and Fred will be located close to each other, such that Fred can assist Alice in self-testing her devices. The same holds for Bob and George.\\

\textbf{Basic inputs and outputs} --
In each round of a QKD protocol each party will receive a share of a global quantum state which is then measured locally on a device that receives a random input and responds with an outcome. In a physical implementation the input typically corresponds to the choice of a measurement base and the output is the result of that measurement, just as in a typical Bell correlation experiment. 
We model measurement devices by sets of positive operator valued measures. In a device-independent setting, these can 
w.l.o.g. assumed to be projective (see e.g. \cite{Schwonnek2021}). Concretely we will write    
$\{M_{a\vert x}\}_a$ for Alice, $\{\tilde{M}_{\tilde{a}\vert \tilde{x}}\}_{\tilde{a}}$ for Fred, $\{N_{b\vert y}\}_b$ for Bob, and $\{\tilde{N}_{\tilde{b}\vert \tilde{y}}\}_{\tilde{b}}$ for George, where $x \in \mathsf{X}$, $\tilde{x} \in \tilde{\mathsf{X}}$, $y \in \mathsf{Y}$, and $\tilde{y} \in \tilde{\mathsf{Y}}$  label the respective input sets. The respective outputs are labelled by  $a\in\mathsf{A}$, $\tilde a\in \tilde{\mathsf{A}}$, $b \in \mathsf{B}$ and $\tilde b \in\tilde{\mathsf{B}}$. 
We assume that  all measurements are sampled from private genuine sources of local randomness within each laboratory. We further assume  that locality loopholes are closed. This means we trust quantum physics and assume a shielding between the local laboratories.\\

\textbf{Switches and sources} --
States are generated by sources. Their distribution to the parties is orchestrated by switches. Their setting, in each round, is controlled by attacker independent random variables. In our model, in contrast to previous works \cite{Tan_2024}, we will simply absorb switches into the sources and merely consider a virtual  source  (see dashed box in \autoref{fig:scheme_two_switches}). It is controlled by random variables $T_A$ and $T_B$, that determine the internal routing and by this the state distribution. In our case, with only two parties on each side, these random variables are binary.  In this abstract view, internal mechanisms, like a repeater or relay needed for connecting two physically distinct sources, are absorbed and do not longer play an explicit role in the modelling task.

What can be seen from the outside is that each setting of the switches $(T_A,T_B)=s$ will lead to a different global state $\omega^s$. In our case we hence have a collection of four different states. We employ the convention that the state created for $(T_A,T_B)=(0,0)$ will be used for key generation, denoted by $\rho:=\omega^{(0,0)}$. Rounds with  $(0,1)$ and $(1,0)$ will be used for local Bell tests. We denote those states by $\sigma:=\omega^{(0,1)}$ and $\tau:=\omega^{(1,0)}$. As for measurements, we assume that the inputs for the switches are sampled from private, genuine sources of local randomness. All other functionalities of the switches could in principle be controlled by Eve.\\

\textbf{System description as universal $C^\ast$-algebra} -- In a device-independent model we avoid fixing Hilbert spaces a priori. A common folklore is to propose an optimization over all states and all measurements in all (finite) Hilbert spaces, in order to capture the worst case of all possible implementations and attacks of an actual physical setting. This object can however be mathematically ill-defined and difficult to handle. Instead, we treat the measurement operators as abstract generators of a universal $C^*$-algebra \cite{Blackadar2006OperatorAlgebras}  that encodes only the algebraic constraints that are imposed by the measurement structure and by causal separations. This object, and optimizations over it, are mathematically consistent. 
Concretely, we treat the measurements $
\{M_{a|x}\}_{a,x},\quad \{\tilde M_{\tilde a|\tilde x}\}_{\tilde a,\tilde x},\quad
\{N_{b|y}\}_{b,y},\quad \{\tilde N_{\tilde b|\tilde y}\}_{\tilde b,\tilde y}$
as abstract symbols and impose the following relations \emph{(PVM relations)}. For each party and each fixed input, the corresponding measurement is a projection-valued measure as
\begin{align}
M_{a|x} = M_{a|x}^\ast = M_{a|x}^2,\;
M_{a|x}M_{a'|x} = 0 \;
\sum_a M_{a|x} = \mathbb I \label{eq:PVM-A}
\end{align}
and analogously for $\tilde M_{\tilde a|\tilde x}$, $N_{b|y}$, and $\tilde N_{\tilde b|\tilde y}$. Operators belonging to different parties commute as \emph{(Separation relations).}
\begin{equation}
[ M_{a|x},\tilde M_{\tilde a|\tilde x} ] =
[ M_{a|x}, N_{b|y} ] =
[ M_{a|x},\tilde N_{\tilde b|\tilde y} ]  = 0,
\label{eq:comm-all}
\end{equation}
and similarly for every pair of generators assigned to distinct parties. These relations encode the no-communication assumption at the operator level.\\

\textbf{Local and global algebras} --
For Alice let $\mathcal A_A$ be the universal unital $C^\ast$-algebra generated by $\{M_{a|x}\}_{a,x}$ subject to Eq.~\eqref{eq:PVM-A}, and define $\mathcal A_F,\mathcal A_B,\mathcal A_G$ analogously for Fred, George, and Bob. The full measurement algebra is then the universal $C^\ast$-algebra generated by all parties with the additional commutation relations from Eq.~\eqref{eq:comm-all}. Equivalently, one may write
\begin{align}
\mathcal A_{\mathrm{all}}
\;\cong\;
\mathcal A_A \otimes_{\max} \mathcal A_F \otimes_{\max} \mathcal A_B \otimes_{\max} \mathcal A_G,
\end{align}
where $\otimes_{\max}$ denotes the maximal $C^\ast$-tensor product. This is precisely the completion that enforces mutual commutativity of the represented subalgebras \cite{Blackadar2006OperatorAlgebras}. 

The beauty of this ansatz is that $\mathcal{A}_{all}$ naturally comes with a well-defined state space, formally given by $\mathcal S(\mathcal A_{\mathrm{all}})
:=\{\omega\in \mathcal A_{\mathrm{all}}^\ast:\ \omega\ge 0,\ \omega(\mathbb I )=1\}$. That is, by all positive linear normalized functionals on $\mathcal A_{\mathrm{all}}$. States on subsystems are now naturally obtained by restricting a global state to an corresponding subalgebra. Accordingly we will, e.g., use the notation $\omega_A:=\omega|_{\mathcal A_A}$ to denote a marginal state of Alice and similar for the other parties.\\

\textbf{Alice and Bob are independent of the switch} -- The basic structural assumption of our model is that the switch positions do not reconfigure the devices in Alice's and Bob's laboratories. Operationally, this means that no measurement acting only on Alice's and Bob's systems  reveals any information about the current switch setting. Equivalently, the reduced states in Alice's and Bob's laboratories are independent of the switch. Accordingly, for the key-generation state $\rho$ and the testing states $\sigma,\tau$, we impose the marginal constraints
\begin{align}\label{eq:marginal_constraint}
  \rho_A = \sigma_A
  \quad \text{and} \quad
  \rho_B = \tau_B .
\end{align}
In particular, any statistic $p(a|x)$ observed by Alice (and analogously for Bob) is identical across the corresponding switch settings.\\

\textbf{The role of Eve} -- How to model Eve in routed/local-test scenarios has caused recurring confusion in the literature. Namely, is Eve an additional explicit party, with her own algebra that must commute with (some of) Alice, Bob, Fred, and George, or is she the abstract environment that purifies the relevant reduced state? Different choices appear implicitly in previous analyses and can lead to ambiguous statements about Eve's ``access'' to the local laboratories \cite{Lobo2024,Tan_2024,CerveroMartn2025}. We adopt a definition that avoids this ambiguity by fixing Eve relative to a state $\rho_{AB}$.

Concretely, let $\rho_{AB}$ be the Alice--Bob marginal of the  key generation state $\rho$. We apply the GNS construction \cite{Blackadar2006OperatorAlgebras} to the pair $(\mathcal A_A\otimes_{\max}\mathcal A_B,\rho_{AB})$ and obtain a representation
\[
\pi_{\rho}:\mathcal A_A\otimes_{\max}\mathcal A_B \to \mathcal B(\mathcal H_{\rho}),
\quad
\rho_{AB}(X)=\langle\Omega_{\rho},\,\pi_{\rho}(X)\Omega_{\rho}\rangle .
\]
We then define the $\rho$-dependent von Neumann algebra of Alice and Bob as
$
\mathcal M_{AB}:=\pi_{\rho}(\mathcal A_A\otimes_{\max}\mathcal A_B)^{\prime\prime}
$
and \emph{Eve's} algebra as the commutant
$\mathcal M_E := \mathcal M_{AB}^{\prime}$.

By construction, Eve commutes with exactly the observables accessible to Alice and Bob in the key-generation experiment. No additional commutation requirements have to be guessed. In particular, Fred and George are not part of the definition of Eve. Their role is to generate test statistics that restrict the admissible global models, i.e.~the family of states compatible with all observed data and the marginal constraint. The quantity to be minimized in a security analysis is the conditional entropy with respect to $\mathcal M_E$ induced by the fixed marginal $\rho_{AB}$. This separation is precisely what keeps the attacker model conservative while avoiding the question of whether Eve has access to Fred and George as independent subsystems \cite{Lobo2024,Tan_2024,cerveromartín2023deviceindependentsecurityquantum}. In our numerics, now coming as a result, we use that this assumption can be strengthened to assuming that Eve also commutes with Fred and George, in the key round, without changing the value of a key rate estimate.\\

\textbf{Asymptotic key rates} -- We consider the asymptotic i.i.d.\ regime and assume that the observed frequencies converge to conditional distributions
$p^s(a,\tilde a,b,\tilde b\,|\,x,\tilde x,y,\tilde y)$ for each switch setting $s=(T_A,T_B)$. 
In our $C^\ast$-model these data impose linear constraints on the corresponding states $\omega^s\in\mathsf S(\mathcal A_{\mathrm{all}})$ via
\begin{equation}\label{eq:lin_constraints}
p^s(a,\tilde a,b,\tilde b\,|\,x,\tilde x,y,\tilde y)
=
\omega^s\!\left(M_{a|x}\,\tilde M_{\tilde a|\tilde x}\,N_{b|y}\,\tilde N_{\tilde b|\tilde y}\right),
\end{equation}
together with the marginal (switch-independence) conditions on Alice and Bob. Let $\mathcal F(p)$ denote the resulting feasible set of collections
$\{\omega^s\}_s$ consistent with all constraints.

For key generation we define Eve relative to $\rho_{AB}$ via the GNS construction as above. Let $K_A$ and $K_B$ denote the classical random variables Alice and Bob obtain by a measurement in a key generation round. We write $H(K_A|E)_{\rho_{AB}}$ for the corresponding conditional von Neumann entropy of Alice conditioned on Eve, and $H(K_A|K_B)_{\rho_{AB}}$ for the conditional entropy
of Alice given Bob (both evaluated for the classical-quantum state generated by the key-map applied to $\rho$).

The Devetak--Winter formula \cite{Devetak2005} then yields an achievable asymptotic secret key rate of
\begin{equation}\label{eq:DW_rate}
r_{key}\ \ge\ \inf_{\{\rho^s\}_s\in \mathcal F(p)}
\Bigl[ H(K_A|E)_{\rho_{AB}}\ -\ H(K_A|K_B)_{\rho_{AB}}\Bigr],
\end{equation}
where the infimum reflects the device-independent (worst-case) nature of the analysis: We minimize over all global models for all switch settings
compatible with the observed statistics, while the entropic terms are computed from the key-generation state $\omega^{00}$ only. The computation of the rate in Eq.~\eqref{eq:DW_rate} can be done using the methods described in \cite{ReliableEstimates_2025}. Details are described later in the Method section. 

Finite size corrections to these rates can be bounded via the toolbox of entropy accumulation theorems, see \cite{ArnonFriedman2019,arqand2024generalizedrenyientropyaccumulation} and references therein. By an argumentation parallel to  \cite{Le_Roy_Deloison_2025}, we can make our approach compatible with those techniques.


\subsection{Relation to previous models} \label{sec:recover}

We next relate our C$^\ast$-algebraic routed model to earlier local-test and routed formulations
\cite{Lim2013,Tan_2024,Le_Roy_Deloison_2025}.
We focus on modelling choices that affect the worst-case optimization and indicate how the corresponding pictures are recovered as specializations of our framework.\\

\textbf{Where the switch lives} -- Routed formulations typically introduce an explicit switch input that routes the flying system to a nearby test device or to the distant device, and assumptions are phrased at the level of this routing element \cite{Le_Roy_Deloison_2025,Tan_2024}.
Local-test formulations use an analogous mode choice between test and key behaviour \cite{Lim2013}.
In our formulation, the switch does not appear as an additional subsystem.
Instead, it labels a family of branch states $\{\omega^s\}_s$ on a fixed measurement algebra, and the intended dependence structure is imposed by linear marginal constraints between branches, such as $\sigma_A=\rho_A$ and $\tau_B=\rho_B$.
This isolates the operational content of switch independence in a way that does not depend on how the switch is implemented and matches the role it plays in routed security analyses \cite{Le_Roy_Deloison_2025}.

An explicit-switch description is recovered as follows.
Let $S$ be a classical register and let $\Omega$ be a state on $\mathcal A_{\mathrm{all}}\otimes \mathbb{C}^{|S|}$ with basis $\{\lvert s\rangle\}_s$.
For $p_s:=\Omega(\mathbb{I}\otimes \lvert s\rangle\langle s\rvert)>0$, define conditional branch states
\[
\omega^s(a) := \frac{\Omega\bigl(a\otimes \lvert s\rangle\langle s\rvert\bigr)}{p_s}
\qquad (a\in\mathcal A_{\mathrm{all}}).
\]
Conversely, given a family $\{\omega^s\}_s$ and a distribution $\{p_s\}_s$, one packages them into a single classical-quantum state by
$\Omega(a\otimes \lvert s\rangle\langle s\rvert):=p_s\,\omega^s(a)$.
Thus, explicit-switch models correspond to working with one joint state and passing to conditional branch states, while switch-independence assumptions become constraints relating branch marginals.\\

\textbf{Enforced structure between subsystems} -- Earlier analyses are typically formulated on concrete Hilbert spaces with an explicit identification of subsystems, for instance by a tensor factorization and by specifying which devices act on which factors in which branch \cite{Lim2013,Tan_2024,Le_Roy_Deloison_2025}.
This fixes a representational layer beyond the operational separation assumptions.
Our framework fixes only what is operationally required, namely the PVM relations and separation relations, at the level of a universal measurement C$^\ast$-algebra.
No tensor-product decomposition is imposed a priori.
Tensor-product models are recovered by restricting to those representations and states of the universal algebra that admit the desired factorization.
Dropping this factorization yields the commuting-operator viewpoint compatible with the SDP relaxations used in \cite{Tan_2024,Le_Roy_Deloison_2025}.\\

\textbf{Modelling Eve} -- In routed and local-test formulations, Eve is usually introduced as an explicit side-information system, often via a purification. In particular, it is common to allow the effective test measurement to act jointly on the routed degrees of freedom together with Eve in the test branch \cite{Le_Roy_Deloison_2025}.
Tan et al.\ \cite{Tan_2024} exploit a related reformulation by absorbing intermediate processing into a single state and allowing the relevant measurements to act on an enlarged system without enforcing a fixed tensor split. This highlights that a fixed a priori identification of Eve with a particular subsystem is not canonical once one considers multiple branches and auxiliary test devices.

We therefore do not postulate a fixed Eve algebra a priori. Instead, Eve is defined from the key marginal $\rho_{AB}$ via the GNS construction as the commutant of the represented Alice--Bob algebra. Any model with a fixed Eve system that provides side information compatible with $\rho_{AB}$ is contained as a special case.
The converse direction is not canonical in a multi-branch optimization, since varying $\rho_{AB}$ changes the representation and hence the commutant.
Fitting all candidates into a single ambient Eve algebra would require an additional global construction that is not unique and introduces extra structure that is irrelevant for the key-rate bound. Our definition avoids this bookkeeping and keeps auxiliary test parties out of Eve's definition, they enter only through feasibility constraints induced by the observed test statistics.\\

\textbf{Recovery of bounds} -- The analytic bound of Lim {\it et al.}\ \cite{Lim2013} is recovered in the assisted Alice setting with a single local CHSH test. 
The key step is to convert the observed local-game score into a bound on the effective overlap parameter controlling the BB84-type phase-error contribution.
This conversion uses robust self-testing input for local games, which can now be integrated into the algebraic computation.

For routed Bell tests, Le Roy-Deloison et al.\ \cite{Le_Roy_Deloison_2025} formulate the model through algebraic constraints capturing the branch structure, including joint-measurability type constraints. These constraints are recovered by restricting our feasible set to the corresponding single-switch instance and matching the same separation pattern.
Our numerical evaluation of the conditional entropy uses a different integral representation than \cite{Le_Roy_Deloison_2025}, which changes only the entropy evaluation method, not the admissible set defined by the algebraic constraints.


\subsection{Numerical analysis of protocols with four parties and two switches} \label{sec:numerics}

A central step in the quantitative analysis of a QKD protocol is a  reliable bound on the raw key rate optimization from Eq.~\eqref{eq:DW_rate}. If the feasible set $\mathcal{F}(\rho)$ depends on more parameters than only a single Bell test score, good analytical bounds are unlikely to exist and numerical methods have to be employed. The  motivation for incorporating this additional data into a key rate computation stems from the potential for certifying higher key rates and thresholds. In the following numerical investigations, we therefore consider the full measurement statistics in three different protocols.  

There are at least two essential technical hurdles to take when it comes to bounding the  optimization in Eq.~\eqref{eq:DW_rate}. Firstly, adopting a common view \cite{tan2021computing,Navascus2007}, the feasible set of an DI-experiment has to take into account \emph{all} states on \emph{all} Hilbert spaces, which seems to be an abhorrently large object. In our language this translates to an  optimization over the state space of an $C^*$-algebra, which is now properly defined but typically still infinite dimensional and numerically inaccessible. Secondly, conditional entropies are highly non linear nor Lipschitz-continuos functionals that are already hard to optimize in the finite-dimensional case. 

Computing key rates in a DI setting  therefore 
remained to be an outstanding problem for quite some time \cite{Pirandola_2020}.
Nevertheless, recent progresses now equips us with toolsets \cite{tan2021computing,Brown_2024,brown2021computing,koßmann2024boundingconditionalvonneumannentropy} for handling this. On the one hand, and somewhat surprisingly, optimizations on $C^*$-algebras turned out to be reliably boundable for functionals that are linear \cite{Navascus2007,Navascus2008,koßmann2023hierarchies} or even of small degree  in the underlying state $\omega$. As a result we know have hierarchies, like the famous NPA \cite{Navascus2007}, of SDP \cite{moran2024uncertainty} computable outer bounds for this problem. On the other hand, we can accompany this with methods for reducing the minimization of the conditional entropy to a sequence on linear problems. In \cite{tan2021computing} a non-commutative Gibbs variation is employed. An integral representation based on Gauß-Radau quadratures is used by \cite{Brown_2024}. The work  \cite{koßmann2024boundingconditionalvonneumannentropy}, which reportedly gives the computationally most efficient bounds and is our method of choice, is based on Frenkel's integral formula \cite{Frenkel2023} for relative entropies. Full details about the SDP formulation of Eq.~\eqref{eq:DW_rate} can be found in the methods section. 

In order to asses and  benchmark a concrete protocol, without performing an actual experiment,  we have to generate/simulate realistic data. We will refer to the virtual implementation of a protocol that produces this data as \textit{honest implementation}. For the following setting we consider that all measurements are performed effectively on qubits. 
Furthermore, we model entangled states between two parties by a two qubit Werner states \cite{werner1989quantum}. We stick to the convention of parametrizing these states by 
\begin{align}
    W(v):=v\Psi^-+(1-v) \mathbb{I}/4,
\end{align}
i.e., by the visibility parameter $v$. This parameter determines a the quantum bit error via the formula $Q=(1-v)/2$. We model the quality of the long-range connection in the honest implementation by $\rho_{AB}=W(v)$. The local states used for self-testing are modelled by $\tau=\sigma=W(v_{loc})$.  The choice of the measurements is protocol specific and discussed in the following subsections. Other works \cite{Le_Roy_Deloison_2025} parametrise imperfections via the detector efficiency $\eta$, whereas we use the visibility $v$ as a more general effective noise parameter that can subsume loss and additional experimental imperfections. The most appropriate choice ultimately depends on the targeted hardware implementation.


\subsubsection{Device-Independent BB84}\label{prot:lobotan}

We now consider a four party generalization of a three party protocol investigated by \cite{Tan_2024,Le_Roy_Deloison_2025,Lim2013}. 
Alice and Bob perform a BB84-type protocol. Fred assists Alice via a local test. George assists Bob via a local test. This protocol can be seen as  symmetric extension of  the protocol of Lim {\it et al.}, which  did not incorporate George as an additional local party on Bob's side.

\textbf{Protocol description} -- \\[0.1cm]
\noindent\fbox{
\parbox{0.93\linewidth}{
\textbf{Protocol 1:} Spot checking, 4 parties, binary inputs and outputs, BB84-key generation. \vspace{0.1cm}
\hrule\vspace{0.1cm}
\textbf{1. Measurements}: 
This step is carried out in $N$ rounds. The inputs to the rounds are drawn randomly and independently \\
For Alice: $P(x=0)=p_A, \;P(x=1)=1-p_A$ \\
For Bob $P(y=0)=p_B, P(y=1)=1-p_B$\\
For Fred: $P(\tilde{x}=0)=p_F, \;P(\tilde{x}=1)=1-p_F$\\
For George: $P(\tilde{y}=0)=p_G, \; P(\tilde{y}=1)=1-p_G$\\
For Switch A: $P(T_A=0)=t_A,\;P(T_A=1)=1-t_A$\\
For Switch B: $P(T_B=0)=t_B,\; P(T_B=1)=1-t_B$ \\[0.2cm]
\textbf{2. Sifting:} All inputs are announced over authenticated public channels. A preshared binary random variable $R$ with distribution $(r,1-r)$ is used to select key generation rounds. From rounds with $x=y=T_A=T_B=R=0$ Alice and Bob select a subset of length $\sim rt_A t_B p_A p_B N$. 
The outputs $(a,b)$ of these rounds form a raw key. The data of all other rounds is kept for parameter estimation.\\[0.2cm]
\textbf{3. Parameter estimation:} The data from the test rounds is used to estimate the distribution $p^s(a,\tilde a,b,\tilde b\,|\,x,\tilde x,y,\tilde y)$ by frequencies. If this estimate diverges from a predefined expected behaviour $p^s_{expect}$ by more than a threshold $\varepsilon_{pe}$ the protocol is aborted. \\[0.2cm]
\textbf{4. \& 5. Error correction and Privacy amplification:} 
One-way error correction from Alice to Bob and privacy amplification is performed. 
}
}\\[0.1cm]

\textbf{Parameter Choice} -- The specific distributions of the random variables $T_A$ and $T_B$ do not influence the asymptotic analysis. In general, the ratio of test rounds to key rounds can be made arbitrary small, i.e., $t_A,t_B\rightarrow 1$ when the total number of rounds $N$ approaches infinity. An analogue argumentation holds for the parameters $p_A$ and $p_B$.
We consider an honest implementation in which all sources of imperfection are modelled by local depolarising noise. We assume that local devices can be decomposed in a perfect device followed by depolarising noise and all channels act as identity followed by depolarising noise. We simulate statistics for this protocol, with $p_F=p_G=1/2$ and Werner states with visibilities $v$ and $v_{loc}$ as described above. By switching to the Schrödinger picture, all influences of the depolarising noise are captured by these two parameters. In the honest implementation, we assume that the ideal measurements of Alice and Bob are optimized for a BB84 type protocol. This is for $x=0$ and $y=0$ they will respectively perform a measurement in the Pauli-$Z$ base on a qubit. For $x=1$ and $y=1$ they will respectively perform a measurement in the Pauli-$X$ base. Fred and George pick their measurements in order to maximize the performance in of a self-test. In the depolarizing noise model employed here, this is equivalent to picking measurements that optimize the value of a CHSH game \cite{Clauser1969}. These are, again assuming qubits, measurements in the Pauli $(Z+ X)/\sqrt{2}$ basis for $\tilde{x}=0$ and respectively $\tilde{y}=0$, such as  $(Z-X)/\sqrt{2}$ for $\tilde{x}=1$ and respectively $\tilde{y}=1$.\\[0.1cm]

 \begin{figure}[ht]
  \centering

\definecolor{Vis0930}{HTML}{377EB8}
\definecolor{Vis0940}{HTML}{4DAF4A}
\definecolor{Vis0950}{HTML}{984EA3}
\definecolor{Vis0960}{HTML}{FF7F00}
\definecolor{Vis0970}{HTML}{E41A1C}
\definecolor{Vis0980}{HTML}{A65628}
\definecolor{Vis0990}{HTML}{F781BF}
\definecolor{Vis0995}{HTML}{999999}
\definecolor{Vis1000}{HTML}{1B9E77}

\begin{tikzpicture}
\begin{axis}[
  scaled x ticks=false,
  xticklabel=\pgfkeys{/pgf/number format/.cd,fixed,precision=3,zerofill}\pgfmathprintnumber{\tick},
  width=\linewidth,
  height=0.8\linewidth,
  xlabel={Visibility $v$},
  ylabel={Key rate},
  xticklabel=\pgfkeys{/pgf/number format/.cd,fixed,precision=2,zerofill}\pgfmathprintnumber{\tick},
  grid=both,
  grid style={dashed,very thin},
  tick align=outside,
  tick style={black},
  xmin=0.785, xmax=1.0,
  ymin=0, ymax=1.0,
  ymajorticks=true,
  scaled ticks=false,
  legend cell align=left,
  legend columns=3,
  legend style={
    at={(0.5,-0.2)},
    anchor=north,
    fill=white, fill opacity=0.85,
    draw opacity=1, text opacity=1,
    font=\small,
  },
]

\addplot+[
  black,
  thick,
  dash pattern=on 1.2pt off 1.6pt,
  line cap=round,
  mark=none,
  domain=0.785:1,
  samples=300,
  forget plot
]
{max(0, 1 - 2*hbin((1-x)/2))};

\IfFileExists{results_routed/keyrate_bb84_withoutD_local_0p930.txt}{%
  \addplot+[thick, solid, mark=none, color=Vis0930]
    table[x index=0, y index=1] {results_routed/keyrate_bb84_withoutD_local_0p930.txt};
  \addlegendentry{$v_{\operatorname{loc}}=0.93$}
  \addplot+[thick, densely dashed, mark=none, color=Vis0930, forget plot]
    table[x index=0, y index=1] {results_routed/keyrate_bb84_local_0p930.txt};
}{}

\IfFileExists{results_routed/keyrate_bb84_withoutD_local_0p940.txt}{%
  \addplot+[thick, solid, mark=none, color=Vis0940]
    table[x index=0, y index=1] {results_routed/keyrate_bb84_withoutD_local_0p940.txt};
  \addlegendentry{$v_{\operatorname{loc}}=0.94$}
  \addplot+[thick, densely dashed, mark=none, color=Vis0940, forget plot]
    table[x index=0, y index=1] {results_routed/keyrate_bb84_local_0p940.txt};
}{}

\IfFileExists{results_routed/keyrate_bb84_withoutD_local_0p960.txt}{%
  \addplot+[thick, solid, mark=none, color=Vis0960]
    table[x index=0, y index=1] {results_routed/keyrate_bb84_withoutD_local_0p960.txt};
  \addlegendentry{$v_{\operatorname{loc}}=0.96$}
  \addplot+[thick, densely dashed, mark=none, color=Vis0960, forget plot]
    table[x index=0, y index=1] {results_routed/keyrate_bb84_local_0p960.txt};
}{}

\IfFileExists{results_routed/keyrate_bb84_withoutD_local_0p970.txt}{%
  \addplot+[thick, solid, mark=none, color=Vis0970]
    table[x index=0, y index=1] {results_routed/keyrate_bb84_withoutD_local_0p970.txt};
  \addlegendentry{$v_{\operatorname{loc}}=0.97$}
  \addplot+[thick, densely dashed, mark=none, color=Vis0970, forget plot]
    table[x index=0, y index=1] {results_routed/keyrate_bb84_local_0p970.txt};
}{}

\IfFileExists{results_routed/keyrate_bb84_withoutD_local_0p990.txt}{%
  \addplot+[thick, solid, mark=none, color=Vis0990]
    table[x index=0, y index=1] {results_routed/keyrate_bb84_withoutD_local_0p990.txt};
  \addlegendentry{$v_{\operatorname{loc}}=0.99$}
  \addplot+[thick, densely dashed, mark=none, color=Vis0990, forget plot]
    table[x index=0, y index=1] {results_routed/keyrate_bb84_local_0p990.txt};
}{}

\IfFileExists{results_routed/keyrate_bb84_withoutD_local_1p000.txt}{%
  \addplot+[thick, solid, mark=none, color=Vis1000]
    table[x index=0, y index=1] {results_routed/keyrate_bb84_withoutD_local_1p000.txt};
  \addlegendentry{$v_{\operatorname{loc}}=1.00$}
  \addplot+[thick, densely dashed, mark=none, color=Vis1000, forget plot]
    table[x index=0, y index=1] {results_routed/keyrate_bb84_local_1p000.txt};
}{}

\end{axis}
\end{tikzpicture}
  \caption{Numerical results obtained by solving Eq.~\eqref{eq:optimization} under the NPA relaxation with \cite{ReliableEstimates_2025}. Curves are shown as a function of the (shared) Werner visibility $v$ (equivalently, for Werner-type BB84 statistics, the corresponding QBER is $Q=(1-v)/2$). Colours label the local visibility parameter. For each $v_{\operatorname{loc}}$\,---\,the local visibility of the self-tests\,---\,the solid curve corresponds to the case without the second switch, while the dashed curve corresponds to applying the relaxation on both parties (two switches). The black dotted curve shows the Shor--Preskill \cite{Shor_2000} asymptotic key-rate formula  $r_{\mathrm{SP}}(Q)=1-2h_2(Q)$ with $Q=(1-v)/2$.}
  \label{fig:keyrate-qber-switches}
\end{figure}

\textbf{Results} -- The results of our simulation are presented in \autoref{fig:keyrate-qber-switches} as a function of $v$ for different values of $v_{loc}$. We compare this four party protocol (dashed lines) to its three party predecessor (solid lines) considered in \cite{Le_Roy_Deloison_2025}. Computations were done up to NPA-level 4 on a desktop computer. Every data point took seconds (!) of computing time. We see that including statistics of a fourth party, effectively used for self-testing Bob, will increase key rates whenever $v_{loc}<1$. That is, in regimes with imperfect local self-test. The strength of this effect, i.e. the relative size of the improvement, will increase with the imperfection of the local tests.


\subsubsection{Device-Independent BB84 with random key bases}\label{prot:randomkey}

In device-dependent protocols, like in modern implementations of the entanglement-based BB84, the base used for key extraction is typically fixed. It is straightforward to check that alternating the key base, as originally proposed by \cite{Bennett_2014}, does not improve key rates but adds a post-selection overhead in the amounts of samples required for key generation. Reasons for this can be traced back to symmetries within the set of optimal attack strategies of Eve. It was then shown in \cite{Schwonnek2021}, that these do not longer appear in a fully device-independent situation. Randomly switching the measurement base used for key extraction during a protocol run can yield a substantial advantage. 
It is therefore striking to investigate if this effect also plays a role for the partially characterized devices considered in this work. 
In the following protocol, we adopt this idea and modify the previous protocol by randomly switching the key bases.\newpage

\textbf{Protocol description} -- \\[0.1cm]
\noindent\fbox{
\parbox{0.93\linewidth}{
\textbf{Protocol 2:} Random key base, Spot checking, 4 parties, binary inputs and outputs .  \vspace{0.0cm}%
\hrule\vspace{0.1cm}%
\textbf{1. Measurements}: 
This step is carried out in $N$ rounds. The inputs to the rounds are drawn randomly and independently \\
For Alice: \hfill$P(x=0)=p_A, \;P(x=1)=1-p_A$ \\
For Bob: \hfill$P(y=0)=p_B, P(y=1)=1-p_B$\\
For Fred: \hfill$P(\tilde{x}=0)=p_F, \;P(\tilde{x}=1)=1-p_F$\\
For George: \hfill$P(\tilde{y}=0)=p_G, \; P(\tilde{y}=1)=1-p_G$\\
For Switch A: \hfill$P(T_A=0)=t_A,\;P(T_A=1)=1-t_A$\\
For Switch B: \hfill$P(T_B=0)=t_B,\; P(T_B=1)=1-t_B$ \\[0.2cm]
\textbf{2. Sifting:} All inputs are announced over authenticated public channels. A preshared binary random variable $R$ with distribution $(r,1-r)$ is used to select key generation rounds. In the key rounds sifting is performed in order to post-select on rounds with matching basis choices of Alice and Bob. This is, for rounds with $T_A=T_B=R=0$ and $x=y=0$ or $x=y=1$ Alice and Bob select a subset of length $\sim rt_A t_Bp_b(p_Bp_A+(1-p_A)(1-p_B)) N$. 
The outputs $(a,b)$ of these rounds form a raw key. The data of all other rounds is kept for parameter estimation.\\[0.2cm]
\textbf{3. Parameter estimation:} The data from the test rounds is used to estimate the distribution $p^s(a,\tilde a,b,\tilde b\,|\,x,\tilde x,y,\tilde y)$ by frequencies. If this estimate diverges from a predefined expected behaviour $p^s_{expect}$ by more than a threshold $\varepsilon_{pe}$ the protocol is aborted. \\[0.2cm]
\textbf{4. \& 5. Error correction and Privacy amplification:} 
One-way error correction from Alice to Bob and privacy amplification is performed. 
}
}\\[0.1cm]

\textbf{Parameter Choice} -- As for the previous protocol, the probability for selecting a switch-setting that leads to a key round, controlled by the the parameters $t_A,t_B$, can be made arbitrary close to $1$ for $N$ approaching the asymptotic regime. 
We again simulate this protocol with data obtained from a qubit implementation with depolarizing noise parametrized by visibilities $v_{loc}$ and $v$. 
In an honest implementation the local measurements for Alice are $X$ and $Z$, for $x=0$ and $x=1$, on a qubit. Fred and George perform measurements in the $(Z+ X)/\sqrt{2}$ basis for $\tilde{x}=0$ and respectively $\tilde{y}=0$, such as $(Z-X)/\sqrt{2}$ for $\tilde{x}=1$ and respectively $\tilde{y}=1$. Bob will also switch between the $X$ and the $Z$ bases for $y=0$ and $y=1$. The probability of switching the key bases is uniform, i.e., we set $p_B=p_A=1/2$. For the original setting without self-testing this was shown to be optimal \cite{schwonnek2020robust}.\\

\begin{figure}[h]
  \centering

 \definecolor{Vis0930}{HTML}{377EB8}
\definecolor{Vis0940}{HTML}{4DAF4A}
\definecolor{Vis0950}{HTML}{984EA3}
\definecolor{Vis0960}{HTML}{FF7F00}
\definecolor{Vis0970}{HTML}{E41A1C}
\definecolor{Vis0980}{HTML}{A65628}
\definecolor{Vis0990}{HTML}{F781BF}
\definecolor{Vis0995}{HTML}{999999}
\definecolor{Vis1000}{HTML}{1B9E77}

\begin{tikzpicture}
\begin{axis}[
  scaled x ticks=false,
  xticklabel=\pgfkeys{/pgf/number format/.cd,fixed,precision=3,zerofill}\pgfmathprintnumber{\tick},
  width=\linewidth,
  height=0.8\linewidth,
  xlabel={Visibility $v$},
  ylabel={Key rate},
  xticklabel=\pgfkeys{/pgf/number format/.cd,fixed,precision=2,zerofill}\pgfmathprintnumber{\tick},
  grid=both,
  grid style={dashed,very thin},
  tick align=outside,
  tick style={black},
  xmin=0.785, xmax=1.0,
  ymin=0, ymax=1.0,
  ymajorticks=true,
  scaled ticks=false,
  legend cell align=left,
  legend columns=3,
  legend style={
    at={(0.5,-0.2)},
    anchor=north,
    fill=white, fill opacity=0.85,
    draw opacity=1, text opacity=1,
    font=\small,
  },
]

\addplot+[
  black,
  thick,
  dash pattern=on 1.2pt off 1.6pt,
  line cap=round,
  mark=none,
  domain=0.785:1,
  samples=300,
  forget plot
]
{max(0, 1 - 2*hbin((1-x)/2))};

\IfFileExists{results_routed/keyrate_bb84_local_0p930.txt}{%
  \addplot+[thick, solid, mark=none, color=Vis0930]
    table[x index=0, y index=1] {results_routed/keyrate_bb84_local_0p930.txt};
  \addlegendentry{$v_{\operatorname{loc}}=0.93$}
  \addplot+[thick, densely dashed, mark=none, color=Vis0930, forget plot]
    table[x index=0, y index=1] {results_routed/keyrate_randomkeybasisbb84_local_0p930.txt};
}{}

\IfFileExists{results_routed/keyrate_bb84_local_0p950.txt}{%
  \addplot+[thick, solid, mark=none, color=Vis0950]
    table[x index=0, y index=1] {results_routed/keyrate_bb84_local_0p950.txt};
  \addlegendentry{$v_{\operatorname{loc}}=0.95$}
  \addplot+[thick, densely dashed, mark=none, color=Vis0950, forget plot]
    table[x index=0, y index=1] {results_routed/keyrate_randomkeybasisbb84_local_0p950.txt};
}{}

\IfFileExists{results_routed/keyrate_bb84_local_0p970.txt}{%
  \addplot+[thick, solid, mark=none, color=Vis0970]
    table[x index=0, y index=1] {results_routed/keyrate_bb84_local_0p970.txt};
  \addlegendentry{$v_{\operatorname{loc}}=0.97$}
  \addplot+[thick, densely dashed, mark=none, color=Vis0970, forget plot]
    table[x index=0, y index=1] {results_routed/keyrate_randomkeybasisbb84_local_0p970.txt};
}{}

\IfFileExists{results_routed/keyrate_bb84_local_0p990.txt}{%
  \addplot+[thick, solid, mark=none, color=Vis0990]
    table[x index=0, y index=1] {results_routed/keyrate_bb84_local_0p990.txt};
  \addlegendentry{$v_{\operatorname{loc}}=0.99$}
  \addplot+[thick, densely dashed, mark=none, color=Vis0990, forget plot]
    table[x index=0, y index=1] {results_routed/keyrate_randomkeybasisbb84_local_0p990.txt};
}{}

\IfFileExists{results_routed/keyrate_bb84_local_1p000.txt}{%
  \addplot+[thick, solid, mark=none, color=Vis1000]
    table[x index=0, y index=1] {results_routed/keyrate_bb84_local_1p000.txt};
  \addlegendentry{$v_{\operatorname{loc}}=1.00$}
  \addplot+[thick, densely dashed, mark=none, color=Vis1000, forget plot]
    table[x index=0, y index=1] {results_routed/keyrate_randomkeybasisbb84_local_1p000.txt};
}{}

\end{axis}
\end{tikzpicture}
  \caption{\label{fig:routed-random-key-basis}
Numerical results for the routed random key basis protocol obtained by solving Eq.~\eqref{eq:optimization}
under an NPA relaxation with \cite{koßmann2024semidefiniteoptimizationquantumrelative}. Key rates are plotted versus the Werner visibility $v$ of the long-range link; colours label
the local visibility $v_{\mathrm{loc}}$.
Solid: one switch. Dashed: two switches.
Dotted: Shor--Preskill rate $r_{\mathrm{SP}}(Q)=1-2h_2(Q)$ with $Q=(1-v)/2$.}
\end{figure}

\textbf{Results} -- The numerical key rates for the routed random key basis protocol are shown in \autoref{fig:routed-random-key-basis}
as a function of the long-link visibility $v$ for several values of the local-test visibility $v_{\mathrm{loc}}$. Computations where done up to NPA-level 4 on a desktop computer. Every data point took minutes of compute time. We consider two versions of this protocols, one with (dashed) and one without George (solid). We observe two effects: First, compared to the fixed base protocol of the previous subsection, using a random key basis yields positive key rates already at lower long-link visibilities in the imperfect self-test regimes shown. Second, incorporating an additional local test on Bob's side (two switches, dashed curves) strictly improves the
certified rates whenever $v_{\mathrm{loc}}<1$. As in \autoref{fig:keyrate-qber-switches}, this gain is strongest away from the ideal self-test regime and vanishes for $v_{\mathrm{loc}}\to 1$, where the rates approach the device-dependent Shor--Preskill benchmark.\newpage


\subsubsection{E91 with self-testing} \label{prot:e91pp}

In an honest implementation of the previous protocols Alice and Bob perform measurements in the $X$ and $Z$ basis, and by this, basically implement a self-testing assisted version of the entanglement based BB84/BBM92 protocol. While we can confirm that this basis choice is optimal in cases of perfect self-test\,---\,namely, we recover the device-dependent Shor-Preskill rate\,---\,it is not clear that this choice is also optimal for regimes with a persistent self-testing error. In regimes with a very bad self-testing performance, this is $v_{loc}$ much smaller than one, we would expect an optimal protocol to perform like the fully device-independent settings \cite{pironio2009device} investigated previously. This limiting case does not seem to have been considered in previous protocol designs and we investigate it in the following. 

For fixed key bases, an optimal protocol implementable by qubits was described by Pironio {\it et al.}~\cite{pironio2009device}. Here Bob has three inputs. The first input $x=0$ labels key generation rounds, whereas inputs $x=1,2$ lead to data used for a CHSH test. We adapt a version of this protocol using four parties and local self-tests.\\

\textbf{Protocol description} -- \\[0.1cm]
\noindent\fbox{
\parbox{0.93\linewidth}{
\textbf{Protocol 3:} Spot checking, 4 parties, binary inputs and outputs for Alice Fred and George, 3 inputs and binary outputs for Bob. \vspace{0.1cm}
\hrule\vspace{0.1cm}
\textbf{1. Measurements}: 
This step is carried out in $N$ rounds. The inputs to the rounds are drawn randomly and independently \\
For Alice: $P(x=0)=p_A, \;P(x=1)=1-p_A$ \\
For Bob: $P(y=0)=p_B, P(y=1)=(1-p_B)q_B$, \\\qquad $P(y=2)=(1-p_B)(1-q_B)$\\
For Fred: $P(\tilde{x}=0)=p_F, \;P(\tilde{x}=1)=1-p_F$\\
For George: $P(\tilde{y}=0)=p_G, \; P(\tilde{y}=1)=1-p_G$\\
For Switch A: $P(T_A=0)=t_A,\;P(T_A=1)=1-t_A$\\
For Switch B: $P(T_B=0)=t_B,\; P(T_B=1)=1-t_B$ \\[0.2cm]
\textbf{2. Sifting:} All inputs are announced over authenticated public channels. A preshared binary random variable $R$ with distribution $(r,1-r)$ is used to select key generation rounds. From rounds with $x=y=T_A=T_B=R=0$ Alice and Bob select a subset of length $\sim rt_A t_B p_A p_B N$. The outputs $(a,b)$ of these rounds form a raw key. The data of all other rounds is kept for parameter estimation.\\[0.2cm]
\textbf{3. Parameter estimation:} The data from the test rounds is used to estimate the distribution $p^s(a,\tilde a,b,\tilde b\,|\,x,\tilde x,y,\tilde y)$ by frequencies. If this estimate diverges from a predefined expected behaviour $p^s_{expect}$ by more than a threshold $\varepsilon_{pe}$ the protocol is aborted. \\[0.2cm]
\textbf{4. \& 5. Error correction and Privacy amplification:} 
One-way error correction from Alice to Bob and privacy amplification is performed. 
}
}\\[0.1cm]
\newpage

\textbf{Parameter Choice} -- As for the previous protocols, the probability for selecting a switch-setting that leads to a key round, controlled by the the parameters $t_A,t_B,p_B,p_A$, can be made arbitrary close to $1$ for $N$ approaching the asymptotic regime. All remaining inputs are sampled uniformly. This is, we set $q_B=p_F=p_G=1/2$. 
We again simulate this protocol with data obtained from a qubit implementation with depolarizing noises parametrized by visibilities $v_{loc}$ and $v$. 
In an honest implementation the local measurements for Alice are $X$ and $Z$, for $x=0$ and $x=1$, on a qubit. Fred and George perform measurements in the $(Z+ X)/\sqrt{2}$ basis for $\tilde{x}=0$ and respectively $\tilde{y}=0$, such as $(Z-X)/\sqrt{2}$ for $\tilde{x}=1$ and respectively $\tilde{y}=1$.
Bob will  measure in the $X$   base for  $y=0$. For  $y=1,2$ he will measure in the $(X\pm Z )/\sqrt{2}$ base.\\[0.1cm]

\begin{figure}
  \centering
\definecolor{Vis0940}{HTML}{377EB8}
\definecolor{Vis0960}{HTML}{4DAF4A}
\definecolor{Vis0980}{HTML}{984EA3}
\definecolor{Vis0990}{HTML}{FF7F00}
\definecolor{Vis0993}{HTML}{E41A1C}
\definecolor{Vis0995}{HTML}{A65628}
\definecolor{Vis0997}{HTML}{F781BF}
\definecolor{Vis0999}{HTML}{999999}
\definecolor{Vis1000}{HTML}{1B9E77}

\begin{tikzpicture}
\begin{axis}[
  scaled x ticks=false,
  width=\linewidth,
  height=0.8\linewidth,
  xlabel={Visibility $v$},
  ylabel={Key rate},
  xticklabel=\pgfkeys{/pgf/number format/.cd,fixed,precision=3,zerofill}\pgfmathprintnumber{\tick},
  grid=both,
  grid style={dashed,very thin},
  tick align=outside,
  tick style={black},
  xmin=0.77, xmax=1.0,
  ymin=0, ymax=1.0,
  ymajorticks=true,
  scaled ticks=false,
  legend cell align=left,
  legend columns=3,
  legend style={
    at={(0.5,-0.2)},
    anchor=north,
    fill=white, fill opacity=0.85,
    draw opacity=1, text opacity=1,
    font=\small,
  },
]

\addplot+[
  black,
  thick,
  dotted,
  line cap=round,
  mark=none,
  domain=0.785:1,
  samples=400,
  forget plot
]
{RSPplus(x)};

\addplot+[
  black,
  thick,
  dash pattern=on 2.2pt off 1.8pt,
  line cap=round,
  mark=none,
  domain=0.785:1,
  samples=400,
  forget plot
]
{RPironioPlus(x)};

\IfFileExists{results_routed/bb84_chsh_rotated_bases_keyrate_table_by_local_local_0p960.txt}{%
  \pgfplotstablegetcolsof{results_routed/bb84_chsh_rotated_bases_keyrate_table_by_local_local_0p960.txt}%
  \pgfmathtruncatemacro{\lastcol}{\pgfplotsretval-1}%
  \addplot+[thick, solid, mark=none, color=Vis0960]
    table[col sep=space, x index=0, y index=\lastcol]
    {results_routed/bb84_chsh_rotated_bases_keyrate_table_by_local_local_0p960.txt};
  \addlegendentry{$v_{\operatorname{loc}}=0.96$}%
}{}

\IfFileExists{results_routed/bb84_chsh_rotated_bases_keyrate_table_by_local_local_0p980.txt}{%
  \pgfplotstablegetcolsof{results_routed/bb84_chsh_rotated_bases_keyrate_table_by_local_local_0p980.txt}%
  \pgfmathtruncatemacro{\lastcol}{\pgfplotsretval-1}%
  \addplot+[thick, solid, mark=none, color=Vis0980]
    table[col sep=space, x index=0, y index=\lastcol]
    {results_routed/bb84_chsh_rotated_bases_keyrate_table_by_local_local_0p980.txt};
  \addlegendentry{$v_{\operatorname{loc}}=0.98$}%
}{}

\IfFileExists{results_routed/bb84_chsh_rotated_bases_keyrate_table_by_local_local_0p990.txt}{%
  \pgfplotstablegetcolsof{results_routed/bb84_chsh_rotated_bases_keyrate_table_by_local_local_0p990.txt}%
  \pgfmathtruncatemacro{\lastcol}{\pgfplotsretval-1}%
  \addplot+[thick, solid, mark=none, color=Vis0990]
    table[col sep=space, x index=0, y index=\lastcol]
    {results_routed/bb84_chsh_rotated_bases_keyrate_table_by_local_local_0p990.txt};
  \addlegendentry{$v_{\operatorname{loc}}=0.99$}%
}{}

\IfFileExists{results_routed/bb84_chsh_rotated_bases_keyrate_table_by_local_local_0p993.txt}{%
  \pgfplotstablegetcolsof{results_routed/bb84_chsh_rotated_bases_keyrate_table_by_local_local_0p993.txt}%
  \pgfmathtruncatemacro{\lastcol}{\pgfplotsretval-1}%
  \addplot+[thick, solid, mark=none, color=Vis0993]
    table[col sep=space, x index=0, y index=\lastcol]
    {results_routed/bb84_chsh_rotated_bases_keyrate_table_by_local_local_0p993.txt};
  \addlegendentry{$v_{\operatorname{loc}}=0.993$}%
}{}

\IfFileExists{results_routed/bb84_chsh_rotated_bases_keyrate_table_by_local_local_0p997.txt}{%
  \pgfplotstablegetcolsof{results_routed/bb84_chsh_rotated_bases_keyrate_table_by_local_local_0p997.txt}%
  \pgfmathtruncatemacro{\lastcol}{\pgfplotsretval-1}%
  \addplot+[thick, solid, mark=none, color=Vis0997]
    table[col sep=space, x index=0, y index=\lastcol]
    {results_routed/bb84_chsh_rotated_bases_keyrate_table_by_local_local_0p997.txt};
  \addlegendentry{$v_{\operatorname{loc}}=0.997$}%
}{}

\IfFileExists{results_routed/bb84_chsh_rotated_bases_keyrate_table_by_local_local_0p999.txt}{%
  \pgfplotstablegetcolsof{results_routed/bb84_chsh_rotated_bases_keyrate_table_by_local_local_0p999.txt}%
  \pgfmathtruncatemacro{\lastcol}{\pgfplotsretval-1}%
  \addplot+[thick, solid, mark=none, color=Vis0999]
    table[col sep=space, x index=0, y index=\lastcol]
    {results_routed/bb84_chsh_rotated_bases_keyrate_table_by_local_local_0p999.txt};
  \addlegendentry{$v_{\operatorname{loc}}=0.999$}%
}{}

\IfFileExists{results_routed/bb84_chsh_rotated_bases_keyrate_table_by_local_local_1p000.txt}{%
  \pgfplotstablegetcolsof{results_routed/bb84_chsh_rotated_bases_keyrate_table_by_local_local_1p000.txt}%
  \pgfmathtruncatemacro{\lastcol}{\pgfplotsretval-1}%
  \addplot+[thick, solid, mark=none, color=Vis1000]
    table[col sep=space, x index=0, y index=\lastcol]
    {results_routed/bb84_chsh_rotated_bases_keyrate_table_by_local_local_1p000.txt};
  \addlegendentry{$v_{\operatorname{loc}}=1.00$}%
}{}

\end{axis}
\end{tikzpicture}

  \caption{\label{fig:interpolation}
Numerical asymptotic key rates for Protocol~3 (four-party $E91^{+}$ spot-checking, adapted from Pironio \emph{et al.} \cite{pironio2009device}) as a function of the long-link Werner visibility $v$.
Coloured curves correspond to different local-test visibilities $v_{\mathrm{loc}}$ (depolarizing noise on the local links), as indicated in the legend.
Black dotted: Device-dependent Shor--Preskill benchmark $r_{\mathrm{SP}}(Q)=1-2h_2(Q)$ with $Q=(1-v)/2$. Black dashed: Qubit-based fully device-independent benchmark from Pironio \emph{et al.}~\cite{pironio2009device}.
}
\end{figure}

\textbf{Results} -- \autoref{fig:interpolation} shows asymptotic key rates for Protocol~3 as a function of the long-link  visibility $v$, for several values of the local-test visibility $v_{\mathrm{loc}}$.
Computations where done up to NPA-level 4 on a desktop computer. Every data point took seconds (!) of compute time. As benchmarks, we include the device-dependent Shor--Preskill rate and the rate for the fully device-independent protocol of \cite{pironio2009device} (black curves).
Across the shown range of $v_{\mathrm{loc}}$, the certified key rate increases monotonically with both $v$ and $v_{\mathrm{loc}}$, and approaches the device-dependent benchmark as $v_{\mathrm{loc}}\to 1$, consistent with the local self-test becoming ideal.
Most importantly, for this protocol  of self-testing with visibility $v_{loc}$ now interpolates between an optimal  device-independent and an optimal device-dependent protocol. 

An interesting observation is that we can approach the Shor--Preskill rate even though Bob, in an honest implementation, does not perform the measurements of the BB84-protocol. An explanation for this is as follows: For a characterized Bob ($v_{\mathrm{loc}}=1$), the observable required for phase-error estimation in the Shor--Preskill bound lies in the linear span of the two tilled CHSH test observables performed Bob in our case. So, the missing complementary basis is already determined by the tested directions. Certifying that Bob's effective measurements satisfy this relation in the uncharacterised setting requires George's local self-test, which fixes Bob's measurement plane (up to the usual isometries).
In this sense, the additional party makes a substantial difference, as it enables the certification needed to recover device-dependent performance in the $v_{\mathrm{loc}}\to 1$ limit.\\

\textbf{Optimization of Bob's test basis} -- An obvious variation of the above protocol would be to further alter the measurement basis of Bob's test measurements ($y=1,2$) in the ideal honest implementation. We employed local optimisation heuristics in order to find better a basis. As it turns out, setting Bob's basis mutually unbiased and rotated by $\pi/4$ with respect to the key base is optimal and is exactly performing the implementation described in \cite{pironio2009device}.


\section{Discussion}

We provide an investigation of device-independent QKD protocols assisted by local self-tests. Our study  extends on several previous works by providing mathematically and operationally consistent model in which we can incorporate architectures with multiple switches and self-tests at different location can be analysed. Concretely, we considered local Bell tests at both endpoints of a long-distance link and numerically compute device-independent key-rate bounds for three new protocols. Our numerical analysis indicates that adding a local certification on Bob's side can strictly improve achievable rates when local tests are imperfect. Furthermore, incorporating random key-basis choices can further lower the zero-key threshold compared to fixed-basis routed variants. In the ideal self-testing limit, the rates approach device-dependent BB84 benchmarks, consistent with the intuition that sufficiently strong local certification mitigates the detection-efficiency bottleneck and leaves the rate primarily limited by standard bit/phase errors.

Several task for future work remain. So far, we focus on the asymptotic i.i.d.~setting. Despite being technically possible, the incorporation of finite-size security is left to a point in time when concrete experimental implementations will be investigated. Our numerical investigations also hinge on simulated data, that draws a pessimistic, but still simplified, picture of an realistic noise profile. As of today, dedicated experiments implementing the two-sided routed architecture studied here are still missing. Closing this gap requires an end-to-end demonstration that simultaneously validates the endpoint switch-independence assumptions, achieves high-quality local-test correlations on both sides, and maintains stable long-link visibility under realistic routing and heralding/relay mechanisms. We hope that the optimistic findings of this work will give momentum to the conception of a proof-of-concept experiment. 

Finally, two natural directions are to design an explicit protocol that combines two-sided routing in Protocol~3 with the random key-basis approach of \cite{schwonnek2020robust} in a unified optimization, and to extend the investigation beyond two self-testing parties per side to assess whether additional local testers yield provable and practically relevant rate improvements.


\section{Methods}

\subsection{Setting}

To illustrate our techniques, we consider one additional party on each side (see \autoref{fig:scheme_two_switches}) and, for simplicity, assign them distinct names -- Alice ($A$) and Fred ($A_F$) on Alice’s side, and Bob ($B$) and George ($B_G$) on Bob’s side. It will become clear from the steps below that nothing prevents a straightforward generalization to arbitrarily many parties on each side. To be concrete, we model the switches as sources that generate states at random. Moreover, the long-range experiment via the entanglement-swapping apparatus is represented by a state $\rho_{AB}$, whereby the local state in Alice’s lab by $\sigma_{AA_F}$ and the local state in Bob’s lab by $\tau_{BB_G}$. The assumption that each source acts as a switch\,---\,sending the second subsystem either to another local party or to the entanglement-swapping apparatus\,---\,justifies the marginal constraints
\begin{align}
  \rho_A = \sigma_A
  \quad \text{and} \quad
  \rho_B = \tau_B .
\end{align}
We assume that the sources have access to private, genuine randomness, and likewise that all measurements are sampled from private, genuine sources of local randomness within each laboratory. We further assume for simplicity that locality loopholes are closed. To extract key from this experiment, we must consider all scenarios compatible with the observed statistics, including all states satisfying Eq.~\eqref{eq:marginal_constraint}, together with projection-valued measures (PVMs) $\{M_{a\vert x}\}_a$ for Alice, $\{\tilde{M}_{\tilde{a}\vert \tilde{x}}\}_{\tilde{a}}$ for Fred, $\{N_{b\vert y}\}_b$ for Bob, and $\{\tilde{N}_{\tilde{b}\vert \tilde{y}}\}_{\tilde{b}}$ for George, where $x \in \mathcal{X}$, $\tilde{x} \in \tilde{\mathcal{X}}$, $y \in \mathcal{Y}$, and $\tilde{y} \in \tilde{\mathcal{Y}}$ and all input and output sets are assumed to be finite.


\subsection{Polynomial optimization}

Employing Eq.~\eqref{eq:marginal_constraint}, the optimization problem for $H(A\vert E)$ reads as (where we again denote $\psi_{ABE}$ as a purification of $\rho_{AB}$)
\begin{equation}\label{eq:optimization}
    \begin{aligned}
        \inf_{(\psi,\sigma,\tau)} \ &H(A\vert E)_{\psi} \\
    \operatorname{s.th.} \ &\psi_{ABE},\sigma_{AA_F},\tau_{BB_G} \ \quad \text{states}\\
    &\rho_A = \sigma_A \quad \text{and} \quad \rho_B = \tau_B, \\
    &\operatorname{tr}[\rho_{AB} M_{a\vert x}N_{b\vert y}] = q_{axby} \\
        &\operatorname{tr}[\sigma_{AA_F}M_{a\vert x}M^\prime_{a^\prime\vert x^\prime}] = p_{a a^\prime x x^\prime} \\
        &\operatorname{tr}[\tau_{BB_G}N_{b\vert y}N^\prime_{b^\prime\vert y^\prime}]  = r_{b b^\prime y y^\prime} \\
        &\sum_{a} M_{a\vert x} = 1, \ \sum_{\tilde{a}} \tilde{M}_{\tilde{a}\vert  \tilde{x}} = 1, \quad x \in \mathcal{X}, \ \tilde{x} \in \tilde{\mathcal{X}} \\
        &\sum_{a} N_{b\vert y} = 1, \ \sum_{\tilde{b}} \tilde{N}_{\tilde{b}\vert  \tilde{y}} = 1,  \quad y \in \mathcal{Y}, \ \tilde{y} \in \tilde{\mathcal{Y}} \\
        &[M_{a\vert x},\tilde{M}_{\tilde{a}\vert \tilde{x}}] = [M_{a\vert x},N_{b\vert y}]  = [M_{a\vert x},\tilde{N}_{\tilde{b}\vert \tilde{y}}] = 0, \\
        &[N_{b\vert y},\tilde{N}_{\tilde{b}\vert \tilde{y}}] = 0, \quad \text{for all} \  a,b,x,\tilde{x},y,\tilde{y}.
    \end{aligned}
\end{equation}
A standard technique for solving Eq.~\eqref{eq:optimization} is the Navascués-Pironio-Acín (NPA) hierarchy, applicable when the objective function is a polynomial in the parties’ operators \cite{Navascus2007,Navascus2008}. In our setting, however, the objective is the conditional von Neumann entropy, which is a non-linear functional involving the operator-valued logarithm and is therefore not directly expressible as a polynomial in the operators. To Eq.~\eqref{eq:optimization}, we employ a technique recently developed by some of us \cite{ReliableEstimates_2025}. A relaxation using the NPA hierarchy of Eq.~\eqref{eq:optimization} usually yields many non-commutative variables, so a careful analysis is beneficial to remain resource-efficient. For this reason, we employ the methods from \cite{ReliableEstimates_2025}, which are more resource-efficient than those previously available \footnote{Specifically, in their most efficient form we only need as many variables as the guessing probability for a bunch of separated NPA programs.}, as they just need half of the number of non-commutative variables as \cite{Brown_2024}.

After applying our techniques from \cite{ReliableEstimates_2025}, we obtain a bona fide polynomial optimization problem, in which the marginal constraints can be written as several moment matrices with equality constraints between their entries. Even though this is feasible for sets of projective measurements with not too many inputs and outputs per party, we describe in the following an observation, yielding a resource-efficient relaxation of Eq.~\eqref{eq:optimization}. For this purpose, let us consider an example, with two inputs and two outputs, which, in an ideal case, can be identified with two projections on a qubit \cite{operator_algebra_self_test2024}. For two projections, one can characterize their equivalence class up to a unitary once the norm of their anti-commutator is known. The idea is simply to rewrite the angle between the projections on the Bloch sphere in terms of the norm of the anti-commutator. In this setting, the equality of the marginals in Eq.~\eqref{eq:marginal_constraint} can therefore be replaced by the value attained by the states on the anti-commutator instead of the whole algebra $\mathcal{A}$. Moreover, since the solvability of Eq.~\eqref{eq:optimization} deteriorates as the number of generators in the local algebras of Alice, Bob, Fred, and George increases, we relax Eq.~\eqref{eq:optimization} into several noncommutative optimization problems involving norm estimates of the commutator and the anti-commutator.


\subsection{Security assumptions}

As usual, we assume no knowledge of the devices’ internal workings beyond that they follow the laws of quantum theory and the assumed space-like separation. However, unlike \cite{Le_Roy_Deloison_2025,Tan_2024} depicted in \autoref{fig:scheme_two_switches}~(a), we must clarify how Eve is modelled. In our setting, Alice (with her partners) and Bob (with his partners) are strictly space-like separated, and Eve is space-like separated from both laboratories; all partners operate in secure labs and are not under Eve’s control. This does not violate the overall assumptions of a DIQKD experiment.

By contrast, in \cite{Le_Roy_Deloison_2025,Tan_2024} Eve has access to the switch and may measure the third party in \autoref{fig:scheme_two_switches}~(a) as well as her own system. A key obstacle to extending that model to arbitrarily many parties\,---\,especially intermediate parties between Alice and Bob\,---\,is that all relevant laboratories must be space-like separated. Allowing Eve simultaneous access to all partners would raise the question of her location. Independent of that, our model fully aligns with standard DIQKD assumptions while placing Alice’s and Bob’s partners in secure, space-like separated laboratories. This choice yields a substantially simpler analysis of the achievable key rate in Eq.~\eqref{eq:optimization}.

To prevent attacks on the source, we explicitly discuss how to model the marginal constraints, since they are crucial. Experimentally, the key assumption is that the sources are completely decoupled from Alice’s and Bob's devices. In particular, Alice’s (respectively Bob's) device must not alter its internal workings when the source sends different types of states. Thus, in any concrete implementation of our protocol, it is essential to argue that this assumption holds. A similar assumption is required in the routed scenario of~\autoref{fig:scheme_two_switches}~(a), cf.~the discussion below~\cite[Figure 1]{Le_Roy_Deloison_2025}.


\subsection{Arbitrarily many parties}

We conclude this method section by explaining how to extend our techniques to arbitrarily many parties. By generalizing the switch so that it can route states to an arbitrary number of parties in \autoref{fig:scheme_two_switches} (b), we can perform correlation experiments between Alice and each of those partners, and analogously for Bob. These additional parties can then be used to probe different types of correlations between Alice’s measurements and those of the corresponding partner. The marginal constraint then generalizes to the fact that all marginals on Alice side equal and all marginals on Bob's side equal.\\



\begin{acknowledgments}
M.B.\ and G.K.\ acknowledge support from the Excellence Cluster -- Matter and Light for Quantum Computing (ML4Q-2) and funding by the European Research Council (ERC Grant Agreement No.\ 948139). R.S.\ is supported by the DFG under Germany's Excellence Strategy - EXC-2123 QuantumFrontiers-2 - 390837967 and SFB 1227 (DQ-mat), the Quantum Valley Lower Saxony, and the BMBF projects ATIQ, SEQUIN, Quics and CBQD. The numerical data is available on GitHub \cite{github_numerics}.
\end{acknowledgments}

\bibliographystyle{apsrev4-2}
\bibliography{main} 


\appendix
\setcounter{theorem}{0}

\begin{widetext}


\section{\texorpdfstring{Comparison with the work of Le-Roy et.~al \cite{Le_Roy_Deloison_2025}}{}}\label{appendix:comparison}

Given a routed Bell-test scenario we aim to deconstruct the ways to construct for inputs $\mathcal{A},\mathcal{B},\mathcal{C}$ outputs $\mathcal{X},\mathcal{Y},\mathcal{Z}$. To be concrete, we consider distributions of the form
\begin{align}
    p(a,b\vert x,y) \quad \text{and} \quad p(a,c\vert x,z), \quad \text{for elements} \quad a,b,c,x,y,z.
\end{align}

In the following we aim to concretely compare a one-sided setting as denoted in \autoref{fig:scheme_two_switches} and Eq.~\eqref{eq:optimization}, which would yield
\begin{enumerate}
    \item[(a)] PVMs $\{A_{a\vert x}\}$ on $\mathcal{H}_A$, $\{B_{b\vert y}\}$ on $\mathcal{H}_B$, $\{T_{c\vert z}\}$ on $\mathcal{H}_T$
    \item[(b)] two states $\rho_{ABT}^{S},\rho_{ABT}^{L}$ with $\rho^{L}_A = \rho_A^{S}$
    \item[(c)] observed statistics
    \begin{align}\label{eq:our_correlations}
        \operatorname{tr} \big[A_{a\vert x} \otimes B_{b\vert y} \otimes 1_T  \rho_{ABT}^{L}\big] &= p(a,b\vert x,y) \\
        \operatorname{tr} \big[A_{a\vert x} \otimes 1_B \otimes T_{c\vert z}  \rho_{ABT}^{S}\big] &= p(a,c\vert x,z)
    \end{align}
\end{enumerate}

In contrast, the model imposed by \cite{Le_Roy_Deloison_2025} would be built out of
\begin{enumerate}
    \item PVMs $\{A_{a\vert x}\}$ on $\mathcal{H}_A$, $\{B_{b\vert y}\}$ on $\mathcal{H}_{\tilde{B}}$, and $\{T_{c\vert z}\}$ on $\mathcal{H}_{\tilde{B}} \otimes \mathcal{H}_{E^{\prime}}$
    \item a pure state $\psi_{A\tilde{B}E^{\prime}}$
    \item statistics given by
    \begin{align}\label{eq:correlations_pironio}
        \bra{\psi_{A\tilde{B}E^{\prime}}} A_{a \vert x} \otimes B_{b\vert y} \otimes 1_{E^{\prime}} \ket{\psi_{A\tilde{B}E^{\prime}}} &= p(a,b\vert x,y) \\
        \bra{\psi_{A\tilde{B}E^{\prime}}} A_{a \vert x} \otimes T_{c\vert z} \ket{\psi_{A\tilde{B}E^{\prime}}} &= p(a,c\vert x,z)
    \end{align}
\end{enumerate}

\begin{theorem}\label{thm:implications_models}
    Given a model of type (a--c) with states $\rho^{L}_{ABT}$ and $\rho^{S}_{ABT}$ producing the statistics
    \begin{align}
      p(a,b \vert x,y)&=\operatorname{tr} \big[(A_{a\vert x}\otimes B_{b\vert y}\otimes \mathbf 1_T) \rho^{L}_{ABT}\big],\\
      p(a,c \vert x,z)&=\operatorname{tr} \big[(A_{a\vert x}\otimes \mathbf 1_B\otimes T_{c\vert z}) \rho^{S}_{ABT}\big],
    \end{align}
    and with matching $A$-marginals $\rho^{L}_A=\rho^{S}_A$, let $\ket{\chi^{L}_{ABTE}}$ be a purification of $\rho^{L}_{ABT}$. Then there exists a model of type (1--3) with a pure state $\ket{\psi_{A\tilde B E^{\prime}}}$ and PVMs on $A$, $\tilde B$, and for each $z$ a PVM $\{\Pi^{(z)}_c\}$ on $\tilde B E^{\prime}$ such that the statistics of Eq.~\eqref{eq:correlations_pironio} coincide with those in Eq.~\eqref{eq:our_correlations} and
    \begin{align}
        H(A \vert E)_{\chi^{L}_{AE}} = H(A \vert E^{\prime})_{\psi_{AE^{\prime}}}.
    \end{align}
\end{theorem}

\begin{proof}
Let $\ket{\chi^{L}_{ABTE}}$ and $\ket{\varphi^{S}_{ABTE_S}}$ be purifications of $\rho^{L}_{ABT}$ and $\rho^{S}_{ABT}$, respectively. Since $\rho^{L}_A=\rho^{S}_A$, Uhlmann’s theorem yields an isometry
\begin{align}
W:\mathcal H_{BTE}\longrightarrow\mathcal H_{BTE_S} \quad \text{with} \quad
(1_A\otimes W)\ket{\chi^{L}}=\ket{\varphi^{S}}.
\end{align}

Define the compressed POVMs on $BTE$:
\begin{align}
\widetilde T_{c\vert z} \coloneqq W^\dagger \big(1_B \otimes T_{c\vert z} \otimes 1_{E_S}\big) W,
\qquad \widetilde T_{c\vert z}\ge 0,\quad \sum_c \widetilde T_{c\vert z}=1_{BTE}.
\end{align}
Then for all $a,c,x,z$,
\begin{align}
\bra{\chi^{L}} A_{a\vert x}\otimes \widetilde T_{c\vert z}\ket{\chi^{L}}
=\bra{\varphi^{S}} A_{a\vert x}\otimes 1_B \otimes T_{c\vert z}\otimes 1_{E_S}\ket{\varphi^{S}}
=\operatorname{tr} \big[(A_{a\vert x}\otimes 1_B \otimes T_{c\vert z}) \rho^{S}_{ABT}\big]
=p(a,c \vert x,z).
\end{align}
By definition of $\ket{\chi^{L}}$ we also have for all $a,b,x,y$,
\begin{align}
\bra{\chi^{L}} A_{a\vert x}\otimes B_{b\vert y}\otimes 1_T\otimes 1_E \ket{\chi^{L}} = p(a,b \vert x,y).
\end{align}
Choose a unitary $U_{BT}:\mathcal H_B \otimes \mathcal H_T\to\mathcal H_{\tilde B}$ and set
\begin{align}
\ket{\psi^{\prime}_{A\tilde B E}}\coloneqq(1_A\otimes U_{BT}\otimes 1_E)\ket{\chi^{L}_{ABTE}},\qquad
\tilde B_{b\vert y}\coloneqq U_{BT}(B_{b\vert y}\otimes 1_T)U_{BT}^\dagger,
\end{align}
\begin{align}
\widehat T_{c\vert z}\coloneqq(U_{BT}\otimes 1_E) \widetilde T_{c\vert z} (U_{BT}\otimes 1_E)^\dagger.
\end{align}
Then $\{\tilde B_{b\vert y}\}$ is a PVM on $\tilde B$, each $\{\widehat T_{c\vert z}\}$ is a POVM on $\tilde B E$, and
\begin{align}
\bra{\psi^{\prime}} A_{a\vert x}\otimes \tilde B_{b\vert y}\otimes 1_E \ket{\psi^{\prime}} = p(a,b \vert x,y),\qquad
\bra{\psi^{\prime}} A_{a\vert x}\otimes \widehat T_{c\vert z}\ket{\psi^{\prime}} = p(a,c \vert x,z).
\end{align}

Set $\mathcal H \coloneqq \mathcal H_{\tilde B}\otimes \mathcal H_E$. For each $z$, apply Naimark’s theorem to the POVM $\{\widehat T_{c\vert z}\}_c$ on $\mathcal H$ to obtain a space $\mathcal{K}_z\simeq\mathbb C^{m_z}$ and the canonical isometry\footnote{It is easy to show that $V_z^\dagger V_z = 1_{\mathcal H}$.}
\begin{align}
V_z:\mathcal{H} \to \mathcal{H} \otimes \mathcal{K}_z,\qquad \ket{\phi} \mapsto \sum_c \sqrt{\widehat T_{c\vert z}} \ket{\phi} \otimes \ket{c}.
\end{align}
Now fix a single ancilla space $\mathcal{K}\coloneqq\bigoplus_z \mathbb C^{m_z}$ with orthonormal basis $\{\ket{z,c}\}_{z,c}$, and define the inclusions $J_z:\mathbb C^{m_z}\hookrightarrow\mathcal K$ by $J_z\ket{c}\coloneqq\ket{z,c}$. Fix the embedding
\begin{align}
V:\mathcal{H} \to \mathcal{H} \otimes \mathcal{K},\qquad \ket{\phi} \mapsto \ket{\phi}\otimes \ket{0},
\end{align}
for some distinguished $\ket{0}\in\mathcal K$. There exists a unitary $W_z$ on $\mathcal H\otimes\mathcal K$ such that
\begin{align}
W_z (1_{\mathcal H}\otimes J_z) V_z = V.
\end{align}
This is since $(1_{\mathcal H}\!\otimes J_z)V_z$ and $V$ are isometries with the same initial space $\mathcal H$, their ranges have equal dimension, so the partial isometry sending one range to the other extends to a unitary on $\mathcal H\otimes\mathcal K$.

Define, for each $z$, the orthogonal projections on $\mathcal H\otimes\mathcal K$ by
\begin{align}
\Pi_{c}^{(z)}\coloneqq W_z \big(1_{\mathcal H}\otimes \ket{z,c}\bra{z,c}\big) W_z^\dagger,\qquad
\Pi_{\perp}^{(z)} \coloneqq 1_{\mathcal H\otimes\mathcal K}-\sum_c \Pi_c^{(z)}.
\end{align}
Fix any distinguished outcome $c_0=c_0(z)$ and absorb the complement into $c_0$ by setting
\begin{align}
\widetilde\Pi_c^{(z)} \coloneqq 
\begin{cases}
\Pi_c^{(z)}, & c\neq c_0,\\
\Pi_{c_0}^{(z)}+\Pi_{\perp}^{(z)}, & c=c_0,
\end{cases}
\qquad\text{so that}\qquad \sum_c \widetilde\Pi_c^{(z)} = 1_{\mathcal H\otimes\mathcal K}.
\end{align}
Using $W_z (1_{\mathcal H}\otimes J_z) V_z = V$ one verifies the compression identity, for all $c$,
\begin{align}
V^\dagger \widetilde\Pi_c^{(z)} V \;=\; V_z^\dagger (1_{\mathcal H}\otimes \ket{c}\bra{c}) V_z \;=\; \widehat{T}_{c\vert z},
\end{align}
where we used $V^\dagger \Pi_{\perp}^{(z)} V=0$. In the following we rename $\widetilde\Pi_c^{(z)}$ back to $\Pi_c^{(z)}$.

Define the pure state
\begin{align}
\ket{\psi_{A\tilde B E^{\prime}}} \coloneqq (1_A\otimes V)\ket{\psi^{\prime}_{A\tilde B E}}
= \ket{\psi^{\prime}_{A\tilde B E}}\otimes \ket{0}_{\mathcal K},
\end{align}
and set $\mathcal H_{E^{\prime}}\coloneqq\mathcal H_E\otimes\mathcal K$. Let the Bob PVM be $\{\tilde B_{b\vert y}\}$ on $\tilde B$, and for each $z$ let the tester PVM be $\{\Pi^{(z)}_c\}$ on $\tilde B E^{\prime}$. Using $V^\dagger \Pi^{(z)}_c V=\widehat T_{c\vert z}$ and the definition of $\ket{\psi}$,
\begin{align}
\bra{\psi} A_{a\vert x}\otimes \tilde B_{b\vert y}\otimes 1_{E^{\prime}} \ket{\psi}
&= \bra{\psi^{\prime}} A_{a\vert x}\otimes \tilde B_{b\vert y}\otimes 1_{E} \ket{\psi^{\prime}} = p(a,b \vert x,y),\\
\bra{\psi} A_{a\vert x}\otimes \Pi^{(z)}_c \ket{\psi}
&= \bra{\psi^{\prime}} A_{a\vert x}\otimes \widehat T_{c\vert z} \ket{\psi^{\prime}} = p(a,c \vert x,z).
\end{align}
Thus the statistics in Eq.~\eqref{eq:correlations_pironio} and Eq.~\eqref{eq:our_correlations} coincide for the constructed (1--3) model.

For the conditional entropy, using duality and invariance under isometries on the conditioning system,
\begin{align}
        H(A\vert E)_{\chi^{L}_{AE}}
        = - H(A\vert BT)_{\chi^{L}}
        = - H(A\vert \tilde{B})_{\psi^{\prime}}
        = H(A\vert E)_{\psi^{\prime}_{AE}}
        = H(A\vert E^{\prime})_{\psi_{AE^{\prime}}}.
\end{align}
\end{proof}
\section{\texorpdfstring{On the Implementation of Eq.~\eqref{eq:optimization}}{}}

In this section we aim to give some details on the implementation of the following optimization program, which is a restatement of Eq.~\eqref{eq:optimization}.
\begin{equation}\label{eq:optimization_appendix}
    \begin{aligned}
        \inf_{(\psi,\sigma,\tau)} \ &H(A\vert E)_{\psi} \\
    \operatorname{s.th.} \ &\psi_{ABE},\sigma_{AA_F},\tau_{BB_G} \ \quad \text{states}\\
    &\rho_A = \sigma_A \quad \text{and} \quad \rho_B = \tau_B, \\
    &\operatorname{tr}[\rho_{AB} M_{a\vert x}N_{b\vert y}] = q_{axby} \\
        &\operatorname{tr}[\sigma_{AA_F}M_{a\vert x}M^\prime_{a^\prime\vert x^\prime}] = p_{a a^\prime x x^\prime} \\
        &\operatorname{tr}[\tau_{BB_G}N_{b\vert y}N^\prime_{b^\prime\vert y^\prime}]  = r_{b b^\prime y y^\prime} \\
        &\sum_{a} M_{a\vert x} = 1, \ \sum_{\tilde{a}} \tilde{M}_{\tilde{a}\vert  \tilde{x}} = 1, \quad x \in \mathcal{X}, \ \tilde{x} \in \tilde{\mathcal{X}} \\
        &\sum_{a} N_{b\vert y} = 1, \ \sum_{\tilde{b}} \tilde{N}_{\tilde{b}\vert  \tilde{y}} = 1,  \quad y \in \mathcal{Y}, \ \tilde{y} \in \tilde{\mathcal{Y}} \\
        &[M_{a\vert x},\tilde{M}_{\tilde{a}\vert \tilde{x}}] = [M_{a\vert x},N_{b\vert y}]  = [M_{a\vert x},\tilde{N}_{\tilde{b}\vert \tilde{y}}] = 0, \\
        &[N_{b\vert y},\tilde{N}_{\tilde{b}\vert \tilde{y}}] = 0, \quad \text{for all} \  a,b,x,\tilde{x},y,\tilde{y}.
    \end{aligned}
\end{equation}
In order to bound the conditional von Neumann entropy we refer to \cite{ReliableEstimates_2025}, which gives a detailed explanation how to solve the program with an efficient discretization. Particularly, we always use the speed-up from \cite[App. E]{ReliableEstimates_2025} (1), which changes the sum over the discretization and the sum infimum. In order to implement the equality conditions $\rho_A = \sigma_A$, we just relax this to all monomials on the certain NPA level. So, we basically add moment-equalities in the relaxed program in order to satisfy these conditions approximatively.  

\end{widetext}

\end{document}